\documentclass[10pt, reqno]{amsart} 

\usepackage[utf8]{inputenc} 
\usepackage{graphicx} 
\graphicspath{{./images/}}
\usepackage{amsopn,amssymb, amsmath, mathtools, amsthm, mathrsfs,xcolor,bm,url,enumitem} 
\usepackage[abbrev]{amsrefs} 
\usepackage[a4paper,lmargin=3cm,rmargin=3cm,tmargin=4cm,bmargin=4cm,marginparwidth=2.8cm,marginparsep=1mm]{geometry} 
\usepackage[bookmarks=true,hyperindex,pdftex,colorlinks,citecolor=blue]{hyperref} 
\usepackage{cleveref}
\usepackage{datetime} 
\usepackage{marginnote}
\usepackage{tabularx, ragged2e} 

\parskip\medskipamount  

\newdateformat{monthyeardate}{\monthname[\THEMONTH], \THEYEAR} 
\newdateformat{monthdayyeardate}{\monthname[\THEMONTH] \THEDAY, \THEYEAR}

\newtheorem{theorem}{Theorem}[section]          
\newtheorem{corollary}[theorem]{Corollary}

\newtheorem{conjecture}[theorem]{Conjecture}
\newtheorem{proposition}[theorem]{Proposition}

\theoremstyle{definition}
\newtheorem{definition}[theorem]{Definition}

\numberwithin{equation}{section}

\newcolumntype{C}{>{\Centering\arraybackslash}X} 

\newcommand{\N}[0]{\mathbb{N}}

\newcommand{\F}[0]{\mathbb{F}}

\newcommand{\Mod}[1]{\ (\mathrm{mod}\ #1)} 

\title{Cubic power functions with optimal second-order differential uniformity}

\author[C. O'Reilly]{Connor O'Reilly}
\address[C. O'Reilly]{Department of Computer Science, Loughborough University, Loughborough, UK}
\email{c.oreilly@lboro.ac.uk}

\author[A. M. Sălăgean]{Ana Sălăgean}
\address[A. M. Sălăgean]{Department of Computer Science, Loughborough University, Loughborough, UK}
\email{a.m.salagean@lboro.ac.uk}

\date{\monthdayyeardate\today} 

\begin{document}

\begin{abstract}
We discuss the second-order differential uniformity of vectorial Boolean functions. The closely related notion of
second-order zero differential uniformity has recently been studied in connection to resistance to the boomerang attack.
We prove that monomial functions with univariate form $x^d$ where $d=2^{2k}+2^k+1$ and $\gcd(k,n)=1$ have optimal second-order differential uniformity. 
Computational results suggest that, up to affine equivalence, these might be the only optimal cubic power functions. We begin work towards generalising such conditions to all monomial functions of algebraic degree 3. We also discuss further questions arising from computational results.
\end{abstract}

\keywords{Second-order differential uniformity, monomials, algebraic degree 3, Boomerang attack}
\subjclass{94D10 (Primary) 11T06, 94A60 (Secondary)}
\maketitle


\section{Introduction and background}

Boolean functions can be understood as mathematical functions that model transformations on bit-strings, and have wide usage across theoretical computer science; we are primarily considering cryptography and coding theory.
Extensive coverage of the application of Boolean functions in cryptography and coding theory is provided in \cites{carlet2021book, cusick2009cryptographic}. Boolean functions have historically been used in symmetric cryptography, for the construction of both stream ciphers 
and block ciphers. In particular, the S-boxes in the Data Encryption Standard (DES) and Advanced Encryption Standard (AES) ciphers can be described using Boolean functions. 
The difficulty lies in choosing  ``secure'' Boolean functions 
in the construction of these ciphers in order to provide resistance against attacks.

One well-studied attack on cryptosystems is the differential attack, introduced by \cite{biham1991differential}. For a Boolean function to be resistant to a differential attack, it must satisfy a criterion equivalent to the outputs of the derivative $D_af$ being as uniformly distributed as possible \cite{carlet2021book}*{Ch. 3.4.1}, or equivalently, having minimal differential uniformity \cites{nyberg1993differential, nyberg1993provable}. This motivates the extensive study of functions with optimal differential uniformity, namely perfect nonlinear and APN functions, as explained in Section~\ref{sec-prelim}.

The Boomerang attack is a variation of a differential attack, initially introduced in \cite{wagner1999boomerang}. This method has an advantage over previously-known differential attacks, in that it beats a bound on the minimum number of texts required to break a cipher in some cases. This disproved a common belief, known in \cite{wagner1999boomerang} as a ``folk theorem'', leading to some aspects in the design of previous ciphers. Since then, many variations of the attack have been introduced, notably \cite{dunkelman2024retracing}.

The main direction to-date in studying resistance to this attack is to generalise the differential uniformity to the second-order derivative of the chosen Boolean function. In \cite{cid2018bct}, the Boomerang Connectivity Table (BCT) was introduced to measure the resistance of an S-box to the differential attack. However, this table does not address the case of Feistel ciphers, and thus later, in \cite{boukerrou2020feistelboomerang}, the Feistel Boomerang Connectivity Table (FBCT) was introduced for this purpose. More recently, there has been a focus on the second-order zero differential spectra and the second-order zero differential uniformity. This notion extends the Feistel Boomerang Connectivity Table to Boolean functions over odd characteristic fields; indeed the definitions of the Feistel Boomerang uniformity and the second-order zero differential uniformity are identical for $p=2$. There have been a number of results on the second-order zero differential spectra of various classes of functions. We provide a full list of studied power functions over $\F_{2^n}$ in Table~\ref{tab:known-2nd-zero-diff-powers}, along with their second-order zero differential uniformity, denoted by $\nabla_f$. We note that, while we do not list them, full descriptions of the second-order zero differential spectra for each table entry can be found in the respective citations provided.

\begin{table}[!ht]
\centering
\setlength\extrarowheight{2pt}
\begin{tabular}{||c|c|c|c||}
        \hline
         $d$ & Condition & $\nabla_f$ & Reference\\
         \hline
         \hline
         $2^n-2$ & $n$ odd, or $n$ even & 0, or 4 & \cite{eddahmani2022explicit}*{Thm.~10} \\
         \hline
         $2^k+1$ & $\gcd(k,n)=1$ & $0$ & \cite{eddahmani2022explicit}*{Thm.~12}\\
         \hline
         $2^{2k}+2^k+1$ & $n=2^{4k}$ & $2^{2k}$ & \cite{eddahmani2022explicit}*{Thm.~14}\\
         \hline
         $2^{m+1}-1$ & $n=2m+1$ or $n=2m$ & 0 or $2^m$ & \cite{man2023indepth}*{Cor.~1}\\
         \hline
         $2^m-1$ & $n=2m+1$ or $n=2m$ & $2^m-4$ & \cite{garg2023secondorder}*{Thm.~4.1}\\
         \hline
         21 & $n$ odd & 4 & \cite{garg2023secondorderapn}*{Thm 4.1}\\
         \hline
         $2^n-2^s$ & $\gcd(n,s+1)=1$, $n-s=3$ & 4 & \cite{garg2023secondorderapn}*{Thm.~4.9}\\
         \hline
         7 & any & 4 & \cite{man2023secondorder}*{Thm.~1}\\
         \hline
         $2^{m+1}+3$ & $n=2m+1$, or $n=2m$ & 4 or $2^m$ & \cite{man2023secondorder}*{Thm.~2}\\
         \hline
         $2^{2k}+2^k+1$ & $\gcd(k,n)=1$ & $\leq 4$ & This paper, Cor.~\ref{cor-zdu-leq-4}\\
         \hline
\end{tabular}
\caption{Power functions $f(x)=x^d$ over $\F_{2^n}$ with known second-order zero differential uniformity.}
\label{tab:known-2nd-zero-diff-powers}
\end{table}

In \cite{aubry2018differential}, the authors introduce the second-order differential uniformity, denoted by $\delta^2(f)$, and compute the second-order differential uniformity of the inverse function, see Table~\ref{tab:known-2nd-diff-powers}. Our approach follows this definition, which we feel most naturally extends the definition of the differential uniformity. Note that $\delta^2(f)\geq \nabla_f$, so any results on the second-order differential uniformity also yield an upper bound on the second-order zero differential uniformity. We also mention the paper \cite{eddahmani2024doublediff}, where the authors define the double differential uniformity, a similar notion defined for Boolean functions from $\F_{2^n}$ to $\F_{2^m}$. The authors then compute values of the double differential uniformity for the inverse, Gold, and Bracken-Leander functions. When $n=m$, the definition of the double differential uniformity is identical to that of the second-order differential uniformity, and thus these results are directly relevant to our work. We summarise results on functions of known second-order differential uniformity, denoted by $\delta^2(f)$, in Table~\ref{tab:known-2nd-diff-powers}. 

\begin{table}[!ht]
\centering
\setlength\extrarowheight{2pt}
\begin{tabular}{||c|c|c|c||}
        \hline
         $d$ & Condition & $\delta^2(f)$ & Reference\\
         \hline
         \hline
         $2^n-2$ & $n\geq 6 $& 8 & \cite{aubry2018differential}*{Prop.~8.1} \\
         \hline
         $2^k+1$ & $\gcd(k,n)=1$ & $2^n$ & \cite{eddahmani2024doublediff}*{Cor.~2}\\
         \hline
         $2^{2k}+2^k+1$ & $n=2^{4k}$ & $2^n$ & \cite{eddahmani2024doublediff}*{Cor.~3}\\
         \hline
         $2^{2k}+2^k+1$ & $\gcd(k,n)=1$ & $4$ & This paper, Thm.~\ref{thm-gcd=1->optimal}\\
         \hline
         $2^{2k}+2^k+1$ & $\gcd(k,n)>1$ & $2^n$ & This paper, Prop.~\ref{prop-gcd>1->u2^n}\\
         \hline
\end{tabular}
\caption{Power functions $f(x)=x^d$ over $\F_{p^n}$ with known second-order differential uniformity}
\label{tab:known-2nd-diff-powers}
\end{table}

Our key result, Theorem~\ref{thm-gcd=1->optimal}, provides a new class of power functions with optimal second-order differential uniformity (the optimal value being 4), namely the functions $f:\F_{2^n}\to\F_{2^n}$ given by $f(x)=x^d$, where $d=2^{2k}+2^k+1$, and $\gcd(k,n)=1$. This also gives an upper bound on their second-order zero differential uniformity 
We also  discuss algebraic degree 3 power functions in the general case, and provide necessary conditions for optimal second-order differential uniformity (see Section~\ref{sec-general}). Finally, we conclude in Section~\ref{sec-comp} by discussing further questions arising from computational results and conjecture that the cubic functions described in Theorem~\ref{thm-gcd=1->optimal} are, up to affine equivalence, the only cubic power functions with optimal second-order differential uniformity.


\section{Preliminaries}\label{sec-prelim}

\subsection*{Notation}

We use $p$ to denote a prime, and $q$ to denote a prime power. We then denote by $\F_{q}$ the finite field containing $q$ elements, and by $\F_{q}^*$ the non-zero elements of $\F_{q}$. Most commonly, we choose $p=2$, $q=2^n$ for some $n\in \N$, and write $\F_{2^n}$. Similarly, we denote by $\F_p^n$ the vector space of dimension $n$ over the field $\F_p$ and, most commonly choosing $p=2$, write $\F_2^n$. We use the convention $\N=\{1,2,3,\dots\}$, i.e. not including $0$. 

\subsection*{Boolean functions}

Given $n,m\in\N$, a Boolean function $f$ is a function $f: \F_2^n\to \F_2$, and a vectorial Boolean function is a function $g: \F_2^n\to \F_2^m$ \cite{carlet2021book}. Note that we can always write $g(x)=(f_1(x), \dots, f_m(x))$ for some Boolean functions $f_1,\dots, f_m$, which we refer to as the coordinate functions of $f$. We commonly refer to vectorial Boolean functions as, simply, Boolean functions, and consider a Boolean function $f: \F_2^n\to \F_2$ as a specific case with $m=1$. In this work, we will always consider Boolean functions with respect to affine equivalence. For the following definitions and further information, we refer to \cite{carlet2021book}.
\begin{definition}
    Given Boolean functions $f,g: \F_2^n \to \F_2^m$, we say that $f,g$ are affine equivalent if there exists some affine automorphims $L_1, L_2$ of $\F_2^n, \F_2^m$ respectively, such that $f=L_2\circ g\circ L_1$.
\end{definition}

We will often refer to the algebraic degree of a Boolean function.
\begin{definition}
    The algebraic normal form, or ANF, of a given Boolean function $f: \F_2^n\to \F_2$ is the unique representation given by:
\begin{align*}
    f(x)=\sum_{I\subseteq\{1,\dots,n\}} a_I \left( \prod_{i\in I} x_i\right) = \sum_{I\subseteq \textit{supp}(x)}a_I
\end{align*}
where $a_I\in \F_2$, and $x=(x_1,\dots, x_n)$.
\end{definition}
\begin{definition}
    The degree of the ANF of a Boolean function $f$ is denoted by $d_{\textit{alg}}f$, and is called the algebraic degree of $f$.
\end{definition}
Note that a Boolean function is known as cubic when it has algebraic degree less than or equal to 3. We now extend these definitions to vectorial Boolean functions.
\begin{definition}
    Given a vectorial Boolean function $f: \F_2^n \to \F_2^m$, the ANF of $f$ is then:
\begin{align}
    f(x) = \sum_{I\subseteq\{1,\dots,n\}} a_I \left( \prod_{i\in I} x_i\right) = \sum_{I\subseteq \textit{supp}(x)}a_I\label{eq:multi}
\end{align}
where $a_I\in \F_2^m$, and $x=(x_1,\dots, x_n)$. The algebraic degree of $f$ is then given by:
\begin{align*}
    d_{\textit{alg}}f=\max\{\lvert I \rvert \ : \ I \subseteq \{1,\dots, n\},\ a_I \neq 0\}
\end{align*}
i.e., the maximal algebraic degree of the coordinate functions of $f$.
\end{definition}

We can equivalently write any (vectorial) Boolean function with $n=m$ uniquely as a map $f: \F_{2^n}\to \F_{2^n}$, in the form:
\begin{align}
    f(x)=\sum_{i=0}^{2^n-1} b_i x^i\label{eq:uni}
\end{align}
where $b_i \in \F_{2^n}$, via the isomorphism between $\F_{2}^n$ and $\F_{2^n}$.

When we write a Boolean function in the form of \eqref{eq:multi}, we say it is in multivariate form, and when we write it in the form of \eqref{eq:uni}, we say it is in univariate form. We will primarily work with Boolean functions in univariate form from here.

\subsection*{Differential uniformity}

We begin with the following key definitions.

\begin{definition}
    Given finite Abelian groups $(G_1, +)$ and $(G_2, +)$, some function $f: G_1 \to G_2$, and $a\in G_1^*$, we define the discrete derivative of $f$ in direction $a$ by:
    \begin{align*}
        D_af(x)= f(x+a)-f(x).
    \end{align*}
Moreover, if $f: G \to G$ for a group $(G,+)$, given $a_i \in G^*$ for $i=1,\dots, k$, we denote repeated differentiation in directions $a_i$ by:
    \begin{align*}
        D_{a_1, a_2, \dots, a_k}^{(k)} f(x) = D_{a_1}D_{a_2}\cdots D_{a_k}f(x).
    \end{align*}
\end{definition}
Note that applying the discrete derivative to a Boolean function always reduces the algebraic degree by at least 1 \cite{lai1994higher}.

We can then define the differential uniformity of a given Boolean function.
\begin{definition}[\cite{nyberg1993differential}*{Sect. 2}]
    Let $G_1$ and $G_2$ be finite Abelian groups. Then $f: G_1 \to G_2$ is called differentially $\delta$-uniform if, for every non-zero $a\in G_1$ and $b\in G_2$, the equation $D_af(x)=b$ has at most $\delta$ solutions in $G_1$. The minimum such $\delta$ is denoted by $\delta_f$ and is called the differential uniformity of $f$.
\end{definition}

We often consider differentiation in the group $(\F_{p^n}, +)$. Functions of optimal differential uniformity are known as perfect nonlinear functions.
\begin{definition}
    We say that a function $f: \F_{p^n} \to \F_{p^n}$ is perfect nonlinear if it is differentially 1-uniform.
\end{definition}

Note that, when $p=2$, for any $a$ we have $D_af(x) = D_af(x+a)$. As such in this case, the function $D_af$ can be at best 2-to-1, so $f$ can be at best differentially 2-uniform.

\begin{definition}[\cite{nyberg1993provable}*{Sect. 3}]
    We say that a function $f: \F_{2^n} \to \F_{2^n}$ is almost perfect nonlinear, or APN, if the function $D_af$ is 2-to-1 for every $a\in \F_{2^n}^*$, i.e., $\lvert\{ D_af(x) \ :\ x\in \F_{2^n}\}\rvert =2^{n-1}$ for all $a\neq 0$.
\end{definition}

We then define the second-order differential uniformity, following from \cite{aubry2018differential} with slight notation changes.

\begin{definition}\label{def-second-diff-unif}
    Given $n\in \N$, let $f: \F_{2^n}\to \F_{2^n}$ be a Boolean function. Then we say that $f$ has second-order differential uniformity $\delta^2(f)$ if $\delta^2(f)$ is the least integer such that, for every $a,b\in \F_{2^n}^*$, $a\neq b$, and every $c\in \F_{2^n}$, the equation $D_{a,b}^{(2)}f(x)=c$ has at most $\delta^2(f)$ solutions, i.e.,
    \begin{align*}
        \delta^2(f)=\max_{a\neq b\in\F_{2^n}^*,\ c\in F_{2^n}}\lvert\{x\in \F_{2^n} \ : \ D^{(2)}_{a,b}f(x)=c\}\rvert.
    \end{align*}
\end{definition}

Note that we always have:
\begin{align*}
    D_{a,b}^{(2)}f(x)=D_{a,b}^{(2)}f(x+a)=D_{a,b}^{(2)}f(x+b)=D_{a,b}^{(2)}f(x+a+b),
\end{align*}
and thus for any Boolean function $f$, we have $\delta^2(f)\geq 4$. Moreover, this gives that $\delta^2(f)$ is always a multiple of 4.
\begin{definition}\label{def-optimal-second-unif}
    Given $n\in \N$, let $f: \F_{2^n}\to \F_{2^n}$ be a Boolean function. We say that $f$ has optimal second-order differential uniformity if $\delta^2(f)=4$.
\end{definition}

Note that the second-order differential uniformity is an (extended) affine invariant \cite{eddahmani2024doublediff}*{Prop.~2}.

The following observation is used widely, e.g. in the proof of \cite{aubry2018differential}*{Prop.~8.1}. It can also be considered as a corollary of \cite{eddahmani2024doublediff}*{Prop.~3}. We state it here for later reference. 
\begin{proposition}\label{prop-a,b-equiv-1,c}
    Let $f: \F_{2^n}\to \F_{2^n}$ be a power function defined by $f(x)=x^d$. Then, for any $a,b\in \F_{2^n}^*$ we have
    $D_{a,b}f(x) = a^d D_{1,\frac{b}{a}}f\left(\frac{x}{a}\right)$. Therefore
    the second-order differential uniformity of $f$ is given by:
    \begin{align*}
        \delta^2(f)=\max_{c\in\F_{2^n}\setminus\{0,1\},\ c'\in F_{2^n}}\lvert\{x\in \F_{2^n} \ : \ D^{(2)}_{1,c}f(x)=c'\}\rvert.
    \end{align*}
\end{proposition}
\begin{proof}
    Given $a,b\in\F_{2^n}^*$, with $a\neq b$, we compute the second derivative $D_{a,b}f$.
    \begin{align*}
        D_{a,b}f(x)&= x^d+(x+a)^d+(x+b)^d+(x+a+b)^d\\
        &= a^d\left(\left(\frac{x}{a}\right)^d+\left(\frac{x}{a} + 1\right)^d+\left(\frac{x}{a} + \frac{b}{a}\right)^d+\left(\frac{x}{a} + 1 + \frac{b}{a}\right)^d\right)\\
        &= a^d D_{1,\frac{b}{a}}f\left(\frac{x}{a}\right).
    \end{align*}
    Relabeling $y=\frac{x}{a}$, and $c=\frac{b}{a}$, we get
    that for any $c'\in \F_{2^n}$:
    \begin{align*}
        \lvert\{x\in \F_{2^n} \ : \ D^{(2)}_{a,b}f(x)=c'\}\rvert = \lvert\{y\in \F_{2^n} \ : \ D^{(2)}_{1,c}f(y)=a^{-d}c'\}\rvert.\tag*{\qedhere}
    \end{align*}
\end{proof}

Finally, we define the second-order zero differential uniformity.

\begin{definition}[\cite{li2022secondorder}]\label{def-second-zero-unif}
    Given $f: \F_{p^n}\to \F_{p^n}$, and $a,b\in \F_{p^n}$, we define the second-order zero differential spectra of $f$ with respect to $a,b$ as:
    \begin{align*}
        \nabla_f(a,b)=\lvert \{x\in \F_{p^n} \ : \ f(x+a+b)-f(x+a)-f(x+b)+f(x)=0\}\rvert.
    \end{align*}
    If $p=2$, we define the second-order zero differential uniformity of $f$ as:
    \begin{align*}
        \nabla_f = \max\{\nabla_f(a,b) \ : \ a \neq b,\ a,b \neq 0\}.
    \end{align*}
    If $p>2$, then we define the second-order zero differential uniformity of $f$ as:
    \begin{align*}
        \nabla_f = \max\{\nabla_f(a,b) \ : \ a,b\neq 0\}.
    \end{align*}
\end{definition}

Note that when $p=2$ we always have $\nabla_f \leq \delta^2(f)$, and similarly to $\delta^2(f)$, $\nabla_f$ is always a multiple of 4.


\section{Power functions with exponent \texorpdfstring{$2^{2k}+2^k+1$}{2/2k+2/k+1}}\label{sec-2k-exp}

We first consider a particular class of power functions, and give results on its second-order differential uniformity in all cases. We note similarity to the well-studied Bracken-Leander functions \cite{bracken2010highlynonlinear}, but without the restriction $n=4k$.
\begin{theorem}\label{thm-gcd=1->optimal}
        Let $k,n\in \N$ satisfy $\gcd(k,n)=1$. Then the Boolean function $f: \F_{2^n} \to \F_{2^n}$ given by $f(x)= x^{2^{2k}+2^k+1}$ has optimal second-order differential uniformity.
\end{theorem}
\begin{proof}
    By Proposition~\ref{prop-a,b-equiv-1,c}, it is equivalent to consider only the derivatives $D_{1,c}^{(2)}f$, where $c\in\F_{2^n}\setminus\{0,1\}$. Let $q=2^k$. We then have:
    \begin{align}
        D_{1,c}^{(2)}f(x) &=x^{q^2+q+1} + (x+c)^{q^2+q+1} + (x+1)^{q^2+q+1} + (x+c+1)^{q^2+q+1}\notag\\
        &= x^{q^2+q+1} + (x+c)(x^{q^2}+c^{q^2})(x^q+c^q)+(x+1)(x^{q^2}+1)(x^q+1)\notag\\
        &\quad +(x+c+1)(x^{q^2}+c^{q^2}+1)(x^q+c^q+1)\notag\\
        &\quad \vdots\notag\\
        &= x^{q^2}(c^q+c) + x^q(c^{q^2}+c)+x(c^{q^2}+c^q)\label{eqn:d2(1,c)f}\\
        &\quad + c^{q^2+q} + c^{q^2+1} + c^{q^2} + c^{q+1} + c^q + c.\notag
    \end{align}
    Suppose there exists $x,y$ such that $D_{1,c}^{(2)}f(x)= D_{1,c}^{(2)}f(y)$. Then defining $z=x+y$ we have:
    \begin{align}
        0 &= z^{q^2}(c^q+c) + z^q(c^{q^2}+c)+z(c^{q^2}+c^q).\label{eqn:g(z)=0}
    \end{align}
    Clearly, we have solutions of \eqref{eqn:g(z)=0} for $z\in\{0,1\}$, so assume $z\notin \{0,1\}$. As $c\notin\{ 0,1\}$ and $\gcd(k,n)=1$, we have $c^{q^2}\neq c^q$, $c^{q^2}\neq c$, and $c^q \neq c$. We then have:
    \begin{align}
        \frac{c^{q^2}+c^q}{c^q+c} &= \frac{(c^q+c)^q}{c^q+c} = (c^q+c)^{q-1},\label{eqn:c-div-1}
    \end{align}
    and thus:
    \begin{align}
        \frac{c^{q^2}+c}{c^q+c} &= \frac{c^{q^2}+c^q+c^q+c}{c^q+c}=\frac{c^{q^2}+c^q}{c^q+c} + 1= (c^q+c)^{q-1} + 1.\label{eqn:c-div-2}
    \end{align}
    Dividing \eqref{eqn:g(z)=0} by $(c^q+c)$, from \eqref{eqn:c-div-1} and \eqref{eqn:c-div-2} we have:
    \begin{align*}
        0 &= z^{q^2} + z^q((c^q+c)^{q-1} + 1)+z(c^q+c)^{q-1}\\
        0 &= z^{q^2} + z^q + z^q(c^q+c)^{q-1} + z(c^q+c)^{q-1}\\
        0 &= (z^q + z)^q + (z^q + z) (c^q+c)^{q-1}\\
        0 &= (z^q + z)^{q-1} + (c^q+c)^{q-1}\\
        \implies 1 &= \left(\frac{z^q + z}{c^q + c}\right)^{q-1}.
    \end{align*}
    and thus:
    \begin{align*}
        z^q + z &= c^q + c\\
        z^q + c^q &= z + c\\
        (z+c)^q &= z + c
    \end{align*}
    which implies $z+c\in \F_{2^k}\cap \F_{2^n}$. Since $\gcd(k,n)=1$, we have $\F_{2^k}\cap \F_{2^n} =\F_2$.
   Thus the only solutions of \eqref{eqn:g(z)=0} are $z\in\{0,1,c,1+c\}$, so $D_{1,c}^{(2)}f$ is 4-to-1. As this holds for all $c\in \F_{2^n}^*$, we have $\delta^2(f)=4$.
\end{proof}

\begin{proposition}\label{prop-gcd>1->u2^n}
        Let $k,n\in \N$ satisfy $\gcd(k,n)>1$. Then the Boolean function $f: \F_{2^n} \to \F_{2^n}$ given by $f(x)= x^{2^{2k}+2^k+1}$ has maximal second-order differential uniformity, i.e., second-order differential uniformity $2^n$.
\end{proposition}
\begin{proof}
    Using the notation from the proof of Theorem~\ref{thm-gcd=1->optimal}, we have from \eqref{eqn:g(z)=0} that:
    \begin{align*}
        0 &= z^{q^2}(c^q+c) + z^q(c^{q^2}+c)+z(c^{q^2}+c^q),
    \end{align*}
    where $q=2^k$. As $\gcd(k,n)=u>1$, we have that $\F_{2^u}$ is a subfield of $\F_{2^n}$. Thus whenever $c\in \F_{2^u}\setminus\{0,1\}$ we have $c^{q^2}=c^q=c$, so choosing any such $c$ we have that \eqref{eqn:g(z)=0} holds for all $z$. As such, $D_{1,c}^{(2)}f(0)= D_{1,c}^{(2)}f(0+z)$ for all $z \in \F_{2^n}$, so $f$ has second-order differential uniformity $2^n$.
\end{proof}
Note that $f$ having second-order uniformity $2^n$ does not provide information on the exact value of $\gcd(k,n)$.
Discussing Table~\ref{tab:known-2nd-diff-powers}, Proposition~\ref{prop-gcd>1->u2^n} contains and expands on the known second-order differential uniformity of the Bracken-Leander function (row 3).

Combining the two results above, we find the following result.
\begin{theorem}\label{thm-2^2k-optimal}
    Given $k,n\in \N$, define the Boolean function $f: \F_{2^n} \to \F_{2^n}$ by $f(x)= x^{2^{2k}+2^k+1}$. Then $f$ has optimal second-order uniformity if and only if $\gcd(k,n)=1$. Otherwise, $f$ has second-order uniformity $2^n$.
\end{theorem}

Note that as $\nabla_f\leq \delta^2(f)$, we also have the following result.
\begin{corollary}\label{cor-zdu-leq-4}
    Let $k,n\in \N$ satisfy $\gcd(k,n)=1$. Then the Boolean function $f: \F_{2^n} \to \F_{2^n}$ given by $f(x)= x^{2^{2k}+2^k+1}$ has second-order zero differential uniformity $\nabla_f\leq 4$.
\end{corollary}


\section{General cubic monomials}\label{sec-general}

We now discuss power functions of algebraic degree 3 in the general case. The following is a necessary condition for optimality resembling the contrapositive of Proposition~\ref{prop-gcd>1->u2^n} used in the proof of Theorem~\ref{thm-2^2k-optimal}. 

\begin{proposition}\label{prop-general->gcd}
    Let $i,j,n\in \N$, $0<i<j<n$, $j\neq 2i$, and define the Boolean function $f: \F_{2^n} \to \F_{2^n}$ given by $f(x)= x^{2^j+2^i+1}$. Suppose $f$ has optimal second-order differential uniformity. Then if $n$ is odd, we have $\gcd (i,n)=\gcd (j,n)=\gcd(j-i,n)=1$. If $n$ is even, then one of these equals to 2, and the others equal to 1.
\end{proposition}
\begin{proof}
    Let $c\in \F_{2^n}^*$ such that $c\neq 1$. Then we have:
    \begin{align*}
        D_{1,c}^{(2)}f(x)&= x^{2^j+2^i+1}+(x+1)^{2^j+2^i+1}+(x+c)^{2^j+2^i+1}+(x+1+c)^{2^j+2^i+1}\\
        &= x^{2^j+2^i+1} + (x+1)(x^{2^i}+1)(x^{2^j}+1)+(x+c)(x^{2^i}+c^{2^i})(x^{2^j}+c^{2^j})\\
        &\quad +(x+1+c)(x^{2^i}+1+c^{2^i+1})(x^{2^j}+1+c^{2^j+1})\\
        &\quad\vdots\\
        &= x(c^{2^j}+c^{2^i})+x^{2^i}(c^{2^j}+c)+x^{2^j}(c^{2^i}+c)\\
        &\quad + (c^{2^j}+c^{2^i+1})+(c^{2^i}+c^{2^j+1})+(c^{2^i+2^j}+c).
    \end{align*}
    Then as $f$ is optimal, for each fixed $x$ there exist exactly four elements $y$ such that $D_{1,c}^{(2)}f(x)= D_{1,c}^{(2)}f(y)$. Defining $z=x+y$, we have:
    \begin{align}
        0&= z(c^{2^j}+c^{2^i})+z^{2^i}(c^{2^j}+c)+z^{2^j}(c^{2^i}+c)\notag\\
        &= z(c^{2^j}+c^{2^i})+z^{2^i}(c^{2^j}+c^{2^i}+c^{2^i}+c)+z^{2^j}(c^{2^i}+c)\notag\\
        &= (z+z^{2^i})(c^{2^j}+c^{2^i})+(z^{2^j}+z^{2^i})(c^{2^i}+c).\label{eq:cond2i}
    \end{align}
    and similarly:
    \begin{align}
        0&= (z+z^{2^j})(c^{2^j}+c^{2^i})+(z^{2^j}+z^{2^i})(c^{2^j}+c), \textrm{ and}\label{eq:cond2j}\\
        0&= (z+z^{2^j})(c^{2^i}+c)+(z+z^{2^i})(c^{2^j}+c).\label{eq:condboth}
    \end{align}
    Firstly, suppose $\gcd (i,n)\geq 3$. Then there would exist $c\in \F_{2^{\gcd(i,n)}}\subseteq \F_{2^n}$, $c\neq 0,1$, such that $c^{2^i}+c= 0$. Then, from Equation \eqref{eq:cond2i}, we have:
    \begin{align*}
        0&= (z+z^{2^i})(c^{2^j}+c^{2^i})
    \end{align*}
    for at most 4 solutions. But there exist at least 8 elements $z\in \F_{2^{\gcd(i,n)}}$, contradicting $f$ being optimal. As such, we must have $\gcd(i,n)\leq 2$. Similarly, by Equation \eqref{eq:cond2j} we must have $\gcd(j,n)\leq 2$, and by Equation \eqref{eq:condboth} we must have $\gcd(j-i,n)\leq 2$. Note then that if $n$ is odd, we have $\gcd(i,n)=\gcd(j,n)=\gcd(j-i,n)=1$. This proves the first statement.

    Assume now that $n$ is even, and that $\gcd(i,n)=2$, i.e. $i$ is even also. Again, there would exist $c\in \F_{2^{\gcd(i,n)}}\subseteq \F_{2^n}$, $c\neq 0,1$, such that $c^{2^i}+c= 0$, and from Equations \eqref{eq:cond2i} and \eqref{eq:condboth}, we have:
    \begin{align*}
        0&= (z+z^{2^i})(c^{2^j}+c^{2^i}) \textrm{, and}\\
        0&= (z+z^{2^i})(c^{2^j}+c).
    \end{align*}
    As $f$ is optimal, we must have $c^{2^{j-i}}\neq c$ and $c^{2^j}\neq c$ for all $c\in \F_{2^{\gcd(i,n)}}$ (otherwise these equations would hold for all $z$ for some $c$). As such, we must have $c\notin \F_{2^j}\cap\F_{2^n}$ and $c\notin \F_{2^{j-i}}\cap\F_{2^n}$. The case $c\notin \F_{2^j}\cap\F_{2^n}$ requires either $\F_{2^j}\not\subseteq \F_{2^n}$ or $\F_{2^{\gcd(i,n)}}\not\subseteq \F_{2^j}$, i.e. either $\gcd(j,n)=1$ or $\gcd(\gcd(i,n),j)=\gcd(i,j,n)=1$. If $\gcd(i,j,n)=1$, we must have that $j$ is odd, and so $\gcd(j,n)\neq 2$, so $\gcd(j,n)=1$. Similarly, from $c\notin \F_{2^{j-i}}\cap\F_{2^n}$ we see that $\gcd(j-i,n)=1$.
    
    Similarly, beginning with the assumption that $\gcd(j,n)=2$ implies that $\gcd(i,n)=1=\gcd(j-i,n)$, and assuming $\gcd(j-i,n)=2$ implies $\gcd(i,n)=1=\gcd(j,n)$.

    Conversely, assume $\gcd(i,n)=1=\gcd(j,n)$. Then $i,j$ are odd, so $j-i$ is even, so $\gcd(j-i,n)=2$. Similarly, $\gcd(i,n)=1=\gcd(j-i,n)$ implies $\gcd(j,n)=2$, and $\gcd(j,n)=1=\gcd(j-i,n)$ implies $\gcd(i,n)=2$. This concludes the proof of the second statement.    
\end{proof}

We note that the converse is generally false. For example, $d=11$ for $n=7,8$ satisfies the result, but was computed to have second-order differential uniformity 8 in both cases.

The next result leads to a sufficient condition for a general power function of algebraic degree 3 to be affine equivalent to a Boolean function of the form discussed in Theorem~\ref{thm-2^2k-optimal}.
\begin{proposition}\label{prop-ijn-equivs}
    Let $i,j,n\in \N$, $0<i<j<n$, and define the function $f: \F_{2^n} \to \F_{2^n}$ given by $f(x)= x^{2^j+2^i+1}$. Then if $j\equiv 2i$, or $i\equiv 2j$, or $i+j\equiv 0 \Mod{n}$, there exists $k$ such that $f$ is affine equivalent to $g(x)= x^{2^{2k}+2^k+1}$.
\end{proposition}
\begin{proof}
    If $j\equiv 2i \Mod{n}$, we have $j=2i$ as $i,j<n$, and so $f=g$ for $k=i$, and we are done. If $i \equiv 2j \Mod{n}$, define $k=j-i$, and define $A(x)=x^{2^{2k}}$. Then:
    \begin{align*}
        A(f(x))&=  \left(x^{2^{j}+2^{i}+1}\right)^{2^{2k}}\\
        &= x^{2^{3j-2i}+2^{2j-i}+2^{2k}}\\
        &=x^{2^k2^{2j-i}+2^{2j-i}+2^{2k}}\\
        &=x^{2^{2k}+2^k+1}\\
        &=g(x).
    \end{align*}
    Note that the fourth equality holds as $2^{2j-i}\equiv 1$ (mod $2^n-1$). Finally, if $i+j\equiv 0 \Mod{n}$, define $k=i$, and define $A(x)=x^{2^{k}}$. Then:
    \begin{align*}
        A(f(x))&=  \left(x^{2^{j}+2^{i}+1}\right)^{2^{k}}\\
        &=x^{2^{i+j}+2^{2i}+2^{i}}\\
        &=x^{2^{2k}+2^k+1}\\
        &=g(x). \tag*{\qedhere}
    \end{align*}
\end{proof}
Together with Theorem~\ref{thm-gcd=1->optimal}, this provides the following classes of Boolean functions with optimal second-order differential uniformity.
\begin{corollary}\label{cor-equiv-opt}
    Let $i,j,n\in \N$, $0<i<j<n$, and define the function $f: \F_{2^n} \to \F_{2^n}$ given by $f(x)= x^{2^j+2^i+1}$. Then if either:
    \begin{enumerate}
        \item $j\equiv 2i\Mod{n}$, with $\gcd(i,n)=1$, or
        \item $2j\equiv i\Mod{n}$, with $\gcd(j-i,n)=1$, or
        \item $j\equiv -i\Mod{n}$, with $\gcd(i,n)=1$,
    \end{enumerate}
    the function $f$ has optimal second-order differential uniformity.
\end{corollary}


\section{Experimental results}\label{sec-comp}
To conclude, we provide topics for further investigation arising from computational data. We computed all optimal power functions of the form $f:\F_{2^n}\to \F_{2^n}$, $f(x)=x^d$, for $4\leq n\leq 20$. Noting that $x^d$ is affine equivalent to $x^{2^id}$ for all $i$, we consider exponents up to this transformation. Discussing the data, we make the following key observations.

All the optimal exponents of algebraic degree 3
were of the form $d=2^{2k}+2^k+1$ for some $k$ satisfying $\gcd(k,n)=1$ (up to the aforementioned transformation). In other words, every computed optimal cubic exponent was of the form described in Theorem~\ref{thm-gcd=1->optimal}. This suggests the following conjecture:
\begin{conjecture}
A cubic power function $f: \F_{2^n} \to \F_{2^n}$ defined by $f(x)=x^d$ has optimal second-order differential uniformity if and only if $d = 2^i(2^{2k} + 2^k +1) \Mod{2^n-1} $ for some $k,i$ with $\gcd(k,n)=1$.
\end{conjecture}

Other than the cubic functions, we found exactly two other optimal exponents, both of them of algebraic degree 4, namely $d=15$ for  $n=5$ and  $d=85$ for $n=10$. Observe that these exponents are of the form $d=2^{3k}+2^{2k}+2^k+1$ for $k=1,2$ respectively. For $k=3$, the exponent of this form would be $d=585$, but this exponent did not exhibit optimal second-order differential uniformity for any $n\le 20$. It would be of interest to determine a class of optimal power functions of algebraic degree~4 or higher.



\bibliography{refs}

\begin{bibdiv}
\begin{biblist}

\bib{aubry2018differential}{article}{
      author={Aubry, Yves},
      author={Herbaut, Fabien},
       title={Differential uniformity and second order derivatives for generic polynomials},
        date={2018},
        ISSN={0022-4049},
     journal={Journal of Pure and Applied Algebra},
      volume={222},
      number={5},
       pages={1095\ndash 1110},
         url={https://www.sciencedirect.com/science/article/pii/S0022404917301305},
}

\bib{biham1991differential}{article}{
      author={Biham, Eli},
      author={Shamir, Adi},
       title={Differential cryptanalysis of {DES}-like cryptosystems},
        date={1991},
     journal={J. Cryptology},
      volume={4},
       pages={3\ndash 72},
}

\bib{boukerrou2020feistelboomerang}{article}{
      author={Boukerrou, Hamid},
      author={Huynh, Paul},
      author={Lallemand, Virginie},
      author={Mandal, Bimal},
      author={Minier, Marine},
       title={On the {Feistel} counterpart of the boomerang connectivity table: Introduction and analysis of the {FBCT}},
        date={2020May},
     journal={IACR Transactions on Symmetric Cryptology},
      volume={2020},
      number={1},
       pages={331–362},
         url={https://tosc.iacr.org/index.php/ToSC/article/view/8568},
}

\bib{bracken2010highlynonlinear}{article}{
      author={Bracken, Carl},
      author={Leander, Gregor},
       title={A highly nonlinear differentially 4 uniform power mapping that permutes fields of even degree},
        date={2010},
        ISSN={1071-5797},
     journal={Finite Fields and Their Applications},
      volume={16},
      number={4},
       pages={231\ndash 242},
         url={https://www.sciencedirect.com/science/article/pii/S1071579710000249},
}

\bib{carlet2021book}{book}{
      author={Carlet, Claude},
       title={Boolean functions for cryptography and coding theory},
   publisher={Cambridge University Press},
        date={2021},
}

\bib{cid2018bct}{inproceedings}{
      author={Cid, Carlos},
      author={Huang, Tao},
      author={Peyrin, Thomas},
      author={Sasaki, Yu},
      author={Song, Ling},
       title={Boomerang connectivity table: A new cryptanalysis tool},
        date={2018},
   booktitle={Advances in cryptology -- eurocrypt 2018},
      editor={Nielsen, Jesper~Buus},
      editor={Rijmen, Vincent},
   publisher={Springer International Publishing},
     address={Cham},
       pages={683\ndash 714},
}

\bib{cusick2009cryptographic}{book}{
      author={Cusick, T.W.},
      author={Stanica, P.},
       title={Cryptographic {Boolean} functions and applications},
   publisher={Elsevier Science},
        date={2009},
        ISBN={9780080952222},
         url={https://books.google.co.uk/books?id=OAkhkLSxxxMC},
}

\bib{dunkelman2024retracing}{article}{
      author={Dunkelman, Orr},
      author={Keller, Nathan},
      author={Ronen, Eyal},
      author={Shamir, Adi},
       title={The retracing boomerang attack, with application to reduced-round aes},
        date={2024Jul},
        ISSN={1432-1378},
     journal={Journal of Cryptology},
      volume={37},
      number={3},
       pages={32},
         url={https://doi.org/10.1007/s00145-024-09512-7},
}

\bib{eddahmani2022explicit}{article}{
      author={Eddahmani, Said},
      author={Mesnager, Sihem},
       title={{Explicit values of the DDT, the BCT, the FBCT, and the FBDT of the inverse, the gold, and the Bracken-Leander S-boxes}},
        date={202205},
     journal={Cryptography and Communications},
      volume={14},
}

\bib{eddahmani2024doublediff}{inproceedings}{
      author={Eddahmani, Said},
      author={Mesnager, Sihem},
       title={On the double differential uniformity of vectorial boolean functions},
        date={2024},
   booktitle={Progress in cryptology - africacrypt 2024},
      editor={Vaudenay, Serge},
      editor={Petit, Christophe},
   publisher={Springer Nature Switzerland},
     address={Cham},
       pages={3\ndash 20},
}

\bib{garg2023secondorderapn}{misc}{
      author={Garg, Kirpa},
      author={Hasan, Sartaj~Ul},
      author={Riera, Constanza},
      author={Stanica, Pantelimon},
       title={The second-order zero differential spectra of some {APN} and other maps over finite fields},
        date={2023},
        note={arXiv: 2310.13775 [cs.IT]},
}

\bib{garg2023secondorder}{misc}{
      author={Garg, Kirpa},
      author={Hasan, Sartaj~Ul},
      author={Riera, Constanza},
      author={Stanica, Pantelimon},
       title={The second-order zero differential spectra of some functions over finite fields},
        date={2023},
        note={arXiv: 2309.04219 [cs.IT]},
}

\bib{lai1994higher}{inproceedings}{
      author={Lai, Xuejia},
       title={Higher order derivatives and differential cryptanalysis},
        date={1994},
   booktitle={Communications and cryptography: Two sides of one tapestry},
   publisher={Springer US},
     address={Boston, MA},
         url={https://doi.org/10.1007/978-1-4615-2694-0_23},
}

\bib{li2022secondorder}{article}{
      author={Li, Xia},
      author={Yue, Qin},
      author={Tang, Deng},
       title={The second-order zero differential spectra of almost perfect nonlinear functions and the inverse function in odd characteristic},
        date={202205},
        ISSN={1936-2447},
     journal={Cryptography Commun.},
      volume={14},
      number={3},
       pages={653–662},
         url={https://doi.org/10.1007/s12095-021-00544-5},
}

\bib{man2023secondorder}{misc}{
      author={Man, Yuying},
      author={Li, Nian},
      author={Xiang, Zejun},
      author={Zeng, Xiangyong},
       title={On the second-order zero differential spectra of some power functions over finite fields},
        date={2023},
        note={arXiv: 2310.18568 [cs.IT]},
}

\bib{man2023indepth}{misc}{
      author={Man, Yuying},
      author={Mesnager, Sihem},
      author={Li, Nian},
      author={Xiang, Zejun},
      author={Zeng, Xiangyong},
       title={In-depth analysis of {S-boxes} over binary finite fields concerning their differential and {Feistel} boomerang differential uniformities},
        date={2023},
        note={arXiv: 2309.01881 [cs.IT]},
}

\bib{nyberg1993differential}{inproceedings}{
      author={Nyberg, Kaisa},
       title={Differentially uniform mappings for cryptography},
        date={1993},
   booktitle={Advances in cryptology - eurocrypt '93, workshop on the theory and application of of cryptographic techniques, lofthus, norway, may 23-27, 1993, proceedings},
      series={Lecture Notes in Computer Science},
      volume={765},
   publisher={Springer},
       pages={55\ndash 64},
}

\bib{nyberg1993provable}{inproceedings}{
      author={Nyberg, Kaisa},
      author={Knudsen, Lars~Ramkilde},
       title={Provable security against differential cryptanalysis},
        date={1993},
   booktitle={Advances in cryptology --- crypto' 92},
   publisher={Springer Berlin Heidelberg},
     address={Berlin, Heidelberg},
       pages={566\ndash 574},
}

\bib{wagner1999boomerang}{inproceedings}{
      author={Wagner, David},
       title={The boomerang attack},
        date={1999},
   booktitle={Fast software encryption},
      editor={Knudsen, Lars},
   publisher={Springer Berlin Heidelberg},
     address={Berlin, Heidelberg},
       pages={156\ndash 170},
}

\end{biblist}
\end{bibdiv}

\end{document}